\documentclass[journal,12pt,draftcls,onecolumn]{IEEEtran}
\IEEEoverridecommandlockouts

\usepackage{amsthm}
\usepackage{amssymb}
\usepackage{amsmath}
\usepackage{amssymb}
\usepackage{epsfig}
\usepackage{epsf}
\usepackage{subfigure}
\usepackage{graphicx}

\usepackage{cite}
\usepackage[numbers,sort&compress]{natbib}

\usepackage{url}
\usepackage[]{authblk}
\usepackage{color}
\usepackage{upgreek}
\usepackage{comment}

\def\BibTeX{{\rm B\kern-.05em{\sc i\kern-.025em b}\kern-.08em
    T\kern-.1667em\lower.7ex\hbox{E}\kern-.125emX}}

\newtheorem{theorem}{Theorem}

\newtheorem{remark}{Remark}

\newcommand{\dl}{\delta}

\newcommand{\msf}{\mathsf}

\newcommand{\lp}{\left(}
\newcommand{\rp}{\right)}

\begin{document}

\title{Topological Content Delivery with\\Feedback and Random Receiver Cache}

\author{
Alireza~Vahid,~\IEEEmembership{Senior~Member,~IEEE}
\thanks{Alireza Vahid is with the Electrical Engineering Department of the University of Colorado Denver, Denver, CO, USA. Email: {\sffamily alireza.vahid@ucdenver.edu}.}
}

\maketitle


\begin{abstract}
We study the problem of content delivery in two-user interference channels with altering topology, random available cache at the receivers, and delayed channel knowledge at the transmitter. We establish a new set of outer-bounds on the achievable rates when each receiver has access to a random fraction of the message intended for the other receiver, and when each transmitter is aware of which part of its own message is known to the unintended receiver. The outer-bounds reveal the significant potential rate boost associated with even a small amount of side-information at each receiver. The key in deriving the bounds is to quantify the baseline entropy that will always become available to the unintended receiver given the altering topology and the already available side-information. We will also present matching achievable rates in certain scenarios and outline the challenges in more general settings.  
\end{abstract}



\section{Introduction}
\label{Section:Introduction_EICwCache}

Smartly populating the available local memory, or cache, can greatly enhance the data throughput and latency in wireless networks~\cite{maddah2015cache,naderializadeh2017fundamental,sengupta2017fog,hachem2018degrees,zhang2017fundamental,shariatpanahi2017multi}. Unfortunately, predetermining what needs to be placed at each user's local cache may not be feasible in practice due to privacy issues, lack of centralized decision-making, and mobility of the wireless nodes. In this work, we focus on the benefit of random receiver cache in wireless systems wherein users overhear a random portion of the signal intended for other users through the shared medium. In particular, we do not rely on a conventional two phase strategy where the first phase is for content placement, and the second is for content delivery. Instead, we investigate how to harness the random signal that has become available to different users in order to enhance network throughput.

To provide fundamental results on the capacity region in such scenarios, we focus on interference channels with altering topology. More specifically, we consider the canonical two-user interference channel where each wireless link may be active or inactive (or down) according to some Markov process, and these processes may be correlated across users. This model has grown popular in recent years as it provides a suitable framework to study intermittent communications in massive machine-type systems and high packet failure rate in mmWave communications. We provide a brief summary of the efforts on this model later in the introduction. The randomness in the available receiver-end cache is generated by independent erasure processes, and the transmitters are aware of which portion of their own messages is available to the unintended receiver. We further assume the transmitters become aware of the network topology with unit delay, a suitable model for mmWave and machine-type communications. As the topology is captured by whether each link is active or not, this latter assumption can be thought of as the delayed channel state information at the transmitter (delayed CSIT) model.

\noindent {\bf Contributions}: We present a new set of outer-bounds on the capacity region of the two-user interference channel with altering topology and channel state feedback. The first step in the derivation of the outer-bounds is to quantify the baseline entropy that will always become available to the unintended user regardless of the communication strategy. In particular, the key is to incorporate the apriori side-information at each receiver's local cache, the altering topology, and the delayed channel feedback into our analysis. Next, the outer-bounds are derived by using a genie-aided argument to convert the channel into a one-sided interference channel and then, the bounds are obtained by applying the baseline entropy inequality discussed above. The outer-bounds of course recover those known previously in the literature for the no-cache and full-cache (when the entire message of each user is available to the other one) scenarios. Interestingly, the outer-bounds suggest even a small amount of side-information may drastically improve the capacity region as we will discuss later in the paper.

We then investigate under what conditions these outer-bounds can be achieved. We provide two sets of conditions. First, we show for ``strong channels'' and ``small cache'' sizes (to be quantified in the main results), we can achieve the sum-capacity with symmetric channel parameters. Second, we identify a subset of these conditions for which the entire outer-bound region can be achieved and thus, characterizing the capacity region in those cases. We will highlight the challenges when these conditions are not met and discuss whether we believe the inner or the outer-bounds need to be improved.

\noindent {\bf  Summary of Results on Interference Channels with Altering Topology}: The interference channel with altering topology or the erasure interference channel (EIC) was first introduced in~\cite{vahid2011interference}, where it was referred to as the ``binary fading'' model, to generalize the erasure channel to incorporate interference from other transmitters. The capacity region of the two-user EIC with output feedback was reported in~\cite{vahid2012binary} followed by a comprehensive set of results covering the capacity region under delayed and instantaneous CSIT with or without output feedback in~\cite{AlirezaBFICDelayed}. The model and the results were shown to be a good representative of mmWave packet communications~\cite{vahid2014communication,mishra2017harnessing,lin2018gaussian} and topological dynamics of wireless networks~\cite{sun2013topological,issa2015two}. The model was also proven valuable in studying the impact of channel correlation~\cite{vahid2016does,vahid2018throughput} and local delayed knowledge~\cite{vahid2015impact,vahid2016two} on the capacity of distributed wireless networks. Interestingly, the capacity region under the no CSIT assumption and arbitrary erasure probabilities remains open, and the best known inner and outer bounds were reported in~\cite{vahid2017binary} with alternative proof in~\cite{lin2020fast,lin2020new}, echoing the famous ``W-curve'' result of~\cite{etkin2008gaussian}. The model has also been used to study the stability region of interference channels~\cite{nassirpour2020throughput} where newer coding techniques compared to the study of the capacity region were reported. Finally, this model was adopted in~\cite{kassab2020space,kassab2020uncoordinated} to investigate the coexistence of critical and non-critical IoT services.

\noindent {\bf Paper Organization}: The rest of the paper is organized as follows. In Section~\ref{Section:Problem_EICwCache}, we present the problem setting and the assumptions we make in this work. Section~\ref{Section:Main_EIC} presents the main contributions and provides further insights and interpretations of the results. The proof of the outer-bounds are presented in Section~\ref{Section:Converse_EIC}, and the achievability region is derived in Section~\ref{Section:Achievability_EIC}. Finally, Section~\ref{Section:Conclusion_EIC} concludes the paper.
 

\section{Problem Formulation}
\label{Section:Problem_EICwCache}

To quantify the impact of available random receiver cache on the capacity region of interference channels with altering topology, we consider the canonical two-user erasure interference channel (EIC) of Figure~\ref{Fig:EICwCache}. The erasure channel model captures altering network topology~\cite{sun2013topological} or packet failure~\cite{vahid2014communication}. In this network, two single-antenna transmitters, $\msf{Tx}_1$ and $\msf{Tx}_2$, wish to transmit two independent messages, $W_1$ and $W_2$, to their corresponding single-antenna receiving terminals, $\msf{Rx}_1$ and $\msf{Rx}_2$, respectively, over $n$ channel uses. 

\begin{figure}[!ht]
\centering
\includegraphics[width = 0.75\columnwidth]{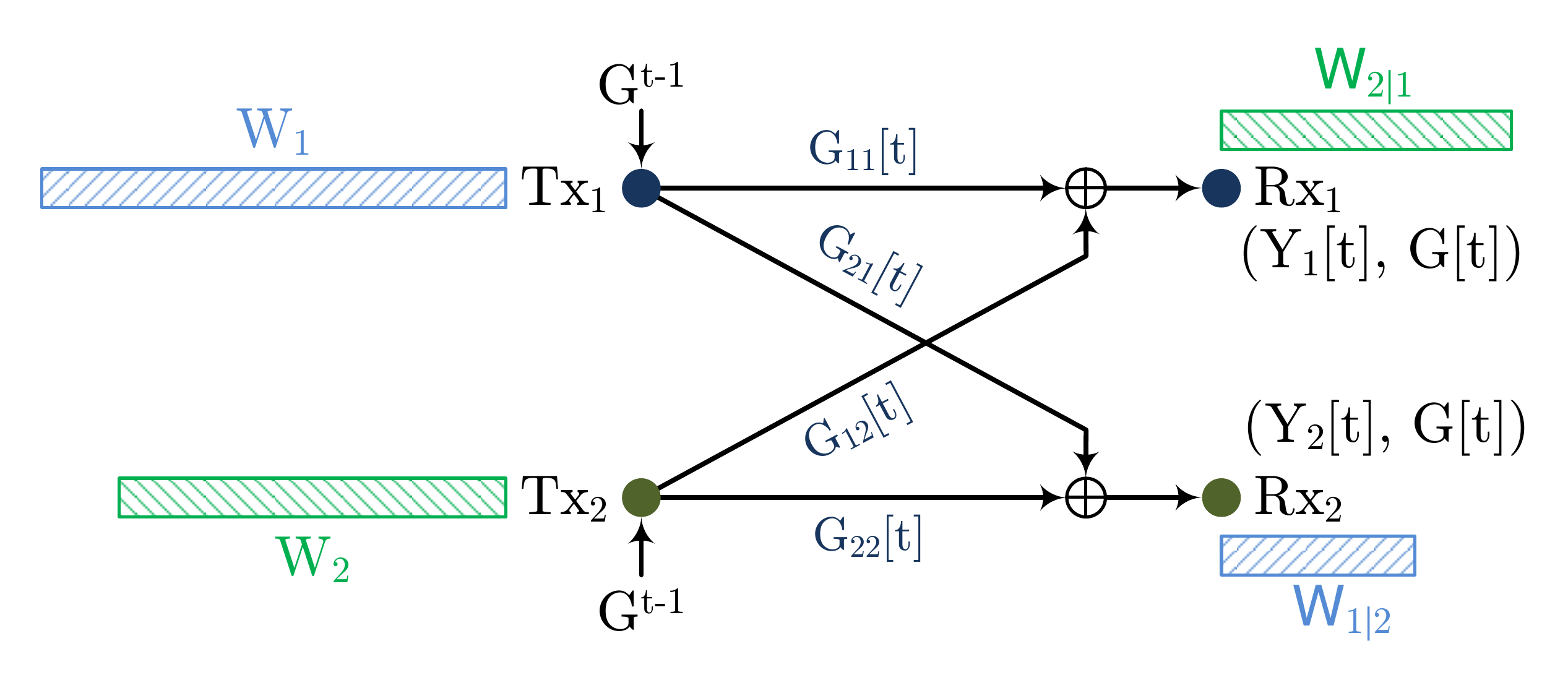}
\caption{The two-user interference channel with altering topology, random local cache at the receivers, and delayed channel knowledge. Each wireless link may be active or down, creating all possible topology configurations for this network.\label{Fig:EICwCache}}
\end{figure}

\noindent \underline{\bf Channel model:} The channel gain from transmitter ${\sf Tx}_j$ to receiver ${\sf Rx}_i$ at time $t$ is denoted by $G_{ij}[t]$, $i,j \in \{1,2\}$. The channel gains are either $0$ or $1$ (\emph{i.e.} $G_{ij}[t] \in \{0,1\}$), and they are distributed as Bernoulli random variables. The channels are assumed to be distributed independently across time but {\it not} necessarily across users. Here, when a channel value is equal to $0$, the receiver does not obtain the corresponding transmitter's signal. We assume: 
\begin{align}
\label{Eq:DeltasGeneral}
& \Pr \left(G_{ij}[t] = 0 \right) = \dl_{ij} \qquad i,j \in \{ 1, 2 \}, \nonumber \\
& \Pr \left(G_{i1}[t] = 0, G_{i2}[t] = 0 \right) = \dl_{{\sf Rx}_i} \qquad i = 1,2, \nonumber \\
& \Pr \left(G_{1j}[t] = 0, G_{2j}[t] = 0 \right) = \dl_{{\sf Tx}_j} \qquad j = 1,2,
\end{align}
for $0 \leq \dl_{ij}, \dl_{{\sf Rx}_i}, \dl_{{\sf Tx}_j} \leq 1$. We note that when the channel gains are distributed independently across users, we have 
\begin{align}
\dl_{{\sf Rx}_i} = \dl_{i1}\dl_{i2}, \qquad \text{and} \qquad \dl_{{\sf Tx}_j} = \dl_{1j}\dl_{2j}. 
\end{align}

\noindent \underline{\bf Input and output signals:} At each time instant $t$, the transmit signal of ${\sf Tx}_j$ is denoted by $X_j[t] \in \{ 0, 1 \}$, and the received signal at ${\sf Rx}_i$ is given by
\begin{equation} 
\label{Eq:ReceivedSignal}
Y_i[t] = G_{ii}[t] X_i[t] \oplus G_{i\bar{i}}[t] X_{\bar{i}}[t], \quad i = 1, 2, \quad \bar{i} \overset{\triangle}= 3 - i,
\end{equation}
where all algebraic operations are in $\mathbb{F}_2$. We note that one could assign a continuous channel gain beyond the binary coefficient and also assume additive noise at the receivers, however, this will not change the fundamental of this problem as was the case in~\cite{sun2013topological,issa2015two}. Further, the results can be easily extended to the case where signals are in $\mathbb{F}_q$ and a correction factor of $\log_2q$ will be added to the inner and outer bounds. 

\begin{remark}
Each point-to-point link in this network is effectively an erasure channel, but instead of representing the output by a symbol in $\{ 0, e, 1\}$, we use a channel gain in the binary field. When the link is equal to $1$ (\emph{i.e.} the link is on), the binary output equals to the input, and when the link is equal to $0$ (\emph{i.e.} the link is off), the output is deterministically zero. Below, we explain that the receiver knows the channel value, and thus, can map the observed zero due to the channel being off to an erasure. This way of modeling enables us to easily describe interference at each receiver as in~\eqref{Eq:ReceivedSignal}.
\end{remark}


\noindent \underline{\bf Channel state information:} We define the channel state information (CSI) at time $t$ to be the quadruple 
\begin{align}
\label{Eq:CSI}
G[t] \overset{\triangle}= \left( G_{11}[t], G_{12}[t], G_{21}[t], G_{22}[t] \right),
\end{align}
and for natural number $k$, we set 
\begin{align}
G^{k} \overset{\triangle}= \left( G[1], G[2], \ldots, G[k] \right)^{\top},
\end{align}
where $G[t]$ is defined in (\ref{Eq:CSI}), and $\left( \cdot \right)^\top$ denotes the transpose operation. Finally, we set 
\begin{align}
G_{ii}^t X_i^t \oplus G_{i\bar{i}}^t X_{\bar{i}}^t \overset{\triangle}= \left[ G_{ii}[1] X_i[1] \oplus G_{i\bar{i}}[1] X_{\bar{i}}[1], \ldots, G_{ii}[t] X_i[t] \oplus G_{i\bar{i}}[t] X_{\bar{i}}[t] \right]^{\top}. 
\end{align} 

\noindent \underline{\bf Messages:} Each message, $W_i$, contains $m_i$ data packets, and we denote the packets for $\msf{Rx}_1$ with $\vec{a} = \left( a_1,a_2,\ldots,a_{m_1} \right)$, and the packets for $\msf{Rx}_2$ with $\vec{b} = \left(b_1,b_2,\ldots,b_{m_2} \right)$. Here, we note each packet is a collection of encoded bits, however, for simplicity and without loss of generality, we assume each packet is in the binary field, and we refer to them as bits. As mentioned earlier, if we assume the packets are in $\mathbb{F}_q$ instead, all that would be needed is a correction factor of $\log_2q$ in the inner and outer bounds.


\noindent \underline{\bf Available CSI at the Transmitters}: In this work, we consider the delayed CSIT model in which at time $t$, each transmitter has the knowledge of the channel state information up to the previous time instant (\emph{i.e.} $G^{t-1}$) as depicted in Figure~\ref{Fig:EICwCache}, and the distribution from which the channel gains are drawn, $t=1,2,\ldots,n$. 

\noindent \underline{\bf Available CSI at the Receivers}: At time instant $t$, ${\sf Rx}_i$ has the its local channel state information up to time $t$ (\emph{i.e.} $G_{ii}^{t}$ and $G_{i\bar{i}}^{t}$), see Figure~\ref{Fig:EICwCache}, and the distribution from which the channel gains are drawn. Each receiver then broadcasts its local CSI which becomes available to all other nodes with unit delay. To make notations simpler, and since receivers only decode the messages at the end of the communication block, we assume both receivers have instantaneous knowledge of the entire CSI. We note that each channel gain in the intermittent (erasure) model captures the success or the failure in delivering a large number of bits in the forward channel, and thus, the feedback overhead is negligible. This also explains why the feedback channel is used to share CSI rather than information about the received signals as in~\cite{AlirezaISIT,vahid2012interference} (\emph{i.e.} the channel output feedback assumption).

\noindent \underline{\bf Random receiver cache:} We assume a random fraction $(1-\epsilon_{i})$ of the bits intended for receiver $\msf{Rx}_{\bar{i}}$ are available at $\msf{Rx}_{{i}}$, $\bar{i} = 3 - i$, and we denote this side information with $W_{\bar{i}|i}$ as in Figure~\ref{Fig:EICwCache}. In particular, we assume that each packet intended for $\msf{Rx}_{\bar{i}}$ becomes available to $\msf{Rx}_i$ according to a Bernoulli $\lp 1 - \epsilon_i \rp$ process distributed independently from all other processes and the messages, and that
\begin{align}
\label{Eq:MessageSplitting}
H\lp W_{\bar{i}|i} \rp = \lp 1 - \epsilon_i \rp H\lp W_{\bar{i}} \rp, \qquad i=1,2.
\end{align}

\begin{remark}
\label{remark:sidechannel}
The assumptions we made on the available side-information at each receiver could also be represented using an erasure side channel. More precisely, we can assume available side-information to receiver $i$ is created through
\begin{align}
E_i[\ell]W_i[\ell], \qquad \ell = 1,2,\ldots, nR_i,
\end{align}
where $W_i[\ell]$ is the $\ell^\mathrm{th}$ bit of message $W_i$, and $E_i[\ell]$ is an i.i.d. Bernoulli $(1-\epsilon_i)$ process independent of all other channel parameters. 
\end{remark}

\noindent \underline{\bf  Transmitter's knowledge of side-information:} We consider the scenario in which the transmitters know exactly what fraction of their {\it own} messages is available to the unintended receiver. From the perspective of Remark~\ref{remark:sidechannel}, each transmitter learns the side channel gains associated with its receiver. 


\noindent \underline{\bf Encoding:} The constraint imposed at the encoding function $f_{i,t}(.) $ at time index $t$ is given by:
\begin{align}
\label{eq_enc_function_EICwCache0}
X_i[t] = f_{i,t}\lp W_i, G^{t-1} \rp, 
\end{align}
however, to highlight the transmitters' knowledge of the available side-information at the unintended receiver, we use the following notation:
\begin{align}
\label{eq_enc_function_EICwCache}
X_i[t] = f_{i,t}\lp W_i, W_{i|\bar{i}}, G^{t-1} \rp, 
\end{align}
where we implicitly assume the knowledge of $\delta_1, \delta_2, \epsilon_1,$ and $\epsilon_2$ is available to each transmitter as side-information. 


\noindent \underline{\bf Decoding:} Each receiver $\msf{Rx}_i$, $i=1,2$, uses a decoding function $\varphi_{i,n}\left( Y_i^n, G^n, W_{\bar{i}|i} \right)$ to get an estimate $\widehat{W}_i$ of $W_i$. An error occurs whenever $\widehat{W}_i \neq W_i$. The average probability of error is given by
\begin{align}
\lambda_{i,n} = \mathbb{E}[P(\widehat{W}_i \neq W_i)],
\end{align}
where the expectation is taken with respect to the random choice of the transmitted messages.

\noindent \underline{\bf Capacity region:} We say that a rate pair $(R_1,R_2)$ is achievable, if there exist block encoders at the transmitters, and block decoders at the receivers, such that $\lambda_{i,n}$ goes to zero as the block length $n$ goes to infinity. The capacity region is the closure of the set of the achievable rate pairs and is denoted by $\mathcal{C}$.


\section{Main Results}
\label{Section:Main_EIC}

In this section, we present the main contributions of this paper and provide some insights and intuitions about the findings.

\subsection{Outer-bounds}

The following theorem establishes a new set of outer-bounds on the capacity region of the two-user interference channel with altering topology and random receiver cache. 

\begin{theorem}[Outer-bounds]
\label{THM:Capacity_Out_EICwCache}
For the two-user interference channel with altering topology, delayed CSIT, and random receiver cache as described in Section~\ref{Section:Problem_EICwCache}, we have
\begin{equation}
\label{Eq:Capacity_Out_BIC_No}
\mathcal{C}  \subseteq \mathcal{C}^\mathrm{out} \equiv
\left\{ \begin{array}{ll}
0 \leq R_i \leq \left( 1 -\dl_{ii} \right), & i=1,2, \\
0 \leq \beta_i R_i + R_{\bar{i}} \leq \left( 1 - \dl_{{\sf Rx}_{\bar{i}}} \right), & i = 1,2.
\end{array} \right.
\end{equation}
where
\begin{align}
\label{Eq:Beta_EICwCache}
\beta_i \overset{\triangle}= \frac{\epsilon_{\bar{i}} \left( 1 - \dl_{\bar{i}i} \right)}{\left( 1 - \dl_{{\sf Tx}_i} \right)}.
\end{align}
\end{theorem}

Before establishing the condition under which these bounds are achievable, we provide some interpretations and insights into these converse bounds. First, we re-write the bounds for when the channels are distributed independently across users, and $\delta_{ij} = \dl, \epsilon_i = \epsilon$, $i,j \in \{ 1,2 \}$. In other words, we focus on the symmetric model with independent links. Then, we can modify the region in \eqref{Eq:Capacity_Out_BIC_No} as
\begin{equation}
\label{Eq:RegionSym}
\left\{ \begin{array}{ll}
0 \leq R_i \leq \left( 1 -\dl \right), & i=1,2, \\
0 \leq \frac{\epsilon}{1+\dl} R_i + R_{\bar{i}} \leq \left( 1 - \dl^2 \right), & i = 1,2.
\end{array} \right.
\end{equation}

\begin{remark}
\label{Remark:Threshold}
From \eqref{Eq:RegionSym}, we conclude that when $\epsilon \leq \delta (1+\delta)$, the region is simply described by $0 \leq R_i \leq \left( 1 - \delta \right)$. We note that this latter expression describes the capacity of two parallel non-interfering point-to-point erasure channels. To put into perspective, when $\delta = 1/2$ and $\epsilon \leq 3/4$ (\emph{i.e.} only $1/4$ of each message is available to the unintended user), the outer-bound matches that of two non-interfering erasure channels.
\end{remark}

\begin{figure}[!ht]
\centering
\includegraphics[width = 0.65\columnwidth]{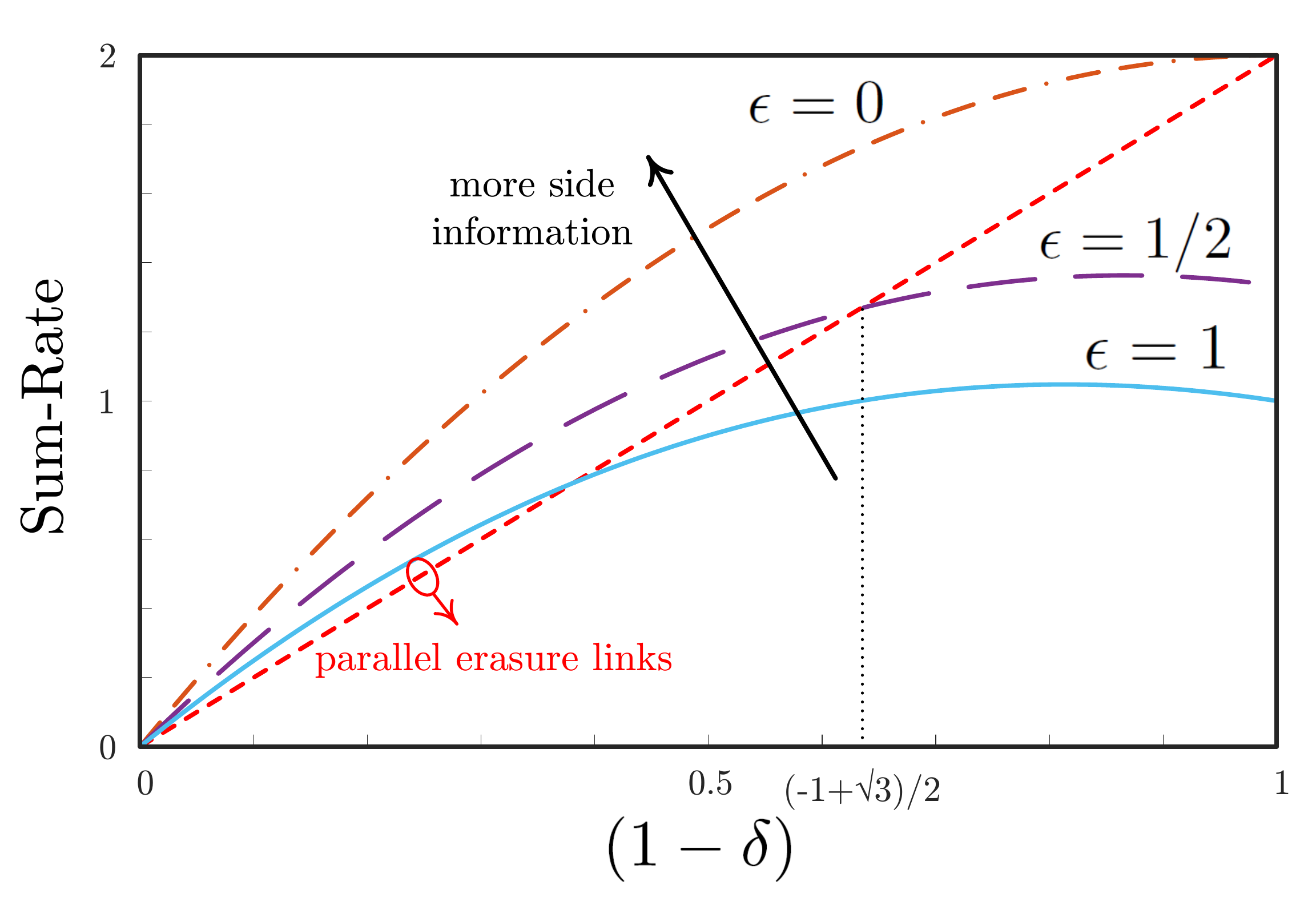}
\caption{The new outer-bounds in \eqref{Eq:NewBounds-EIC-sym} for $\epsilon \in \{ 1, 1/2, 0\}$ as well as the parallel (non-interfering) sum-rate bound as a red dashed line. Note that for convenience, the x-axis represents $(1 - \dl)$, which the probability of each link being active.\label{Fig:Sum-Rates}}
\end{figure}


Figure~\ref{Fig:Sum-Rates} depicts the parallel (non-interfering) sum-rate bound of $2(1 - \dl)$ as well as the sum-rate outer-bound obtained by intersecting  
\begin{align}
\label{Eq:NewBounds-EIC-sym}
\frac{\epsilon}{1+\dl} R_i + R_{\bar{i}} \leq \left( 1 - \dl^2 \right), \qquad i = 1,2,
\end{align}
of \eqref{Eq:RegionSym}. For convenience, the x-axis represents $(1 - \dl)$, which the probability of each link being active. As noted in Remark~\ref{Remark:Threshold} (and also observed in~\cite{AlirezaBFICDelayed} for the no side-information scenario), depending on the value of $\epsilon$ and $\dl$, the sum-rate might be dominated by either of these bounds. For instance, for $\epsilon = 1/2$, when $(1-\dl) \leq (-1+\sqrt{3})/2$, the parallel erasure bounds are dominant and when $(1-\dl) \geq (-1+\sqrt{3})/2$, the bounds in \eqref{Eq:NewBounds-EIC-sym} are dominant. We also note that $\epsilon = 0$ corresponds to the scenario in which the entire message of each user is available to the unintended user, and thus, $2(1 - \dl)$ is easily achievable. On the other end, $\epsilon = 1$ is the scenario with no side-information, and the results recover the region in~\cite{AlirezaBFICDelayed} as expected.


\begin{figure}[!ht]
\centering
\includegraphics[width = 0.4\columnwidth]{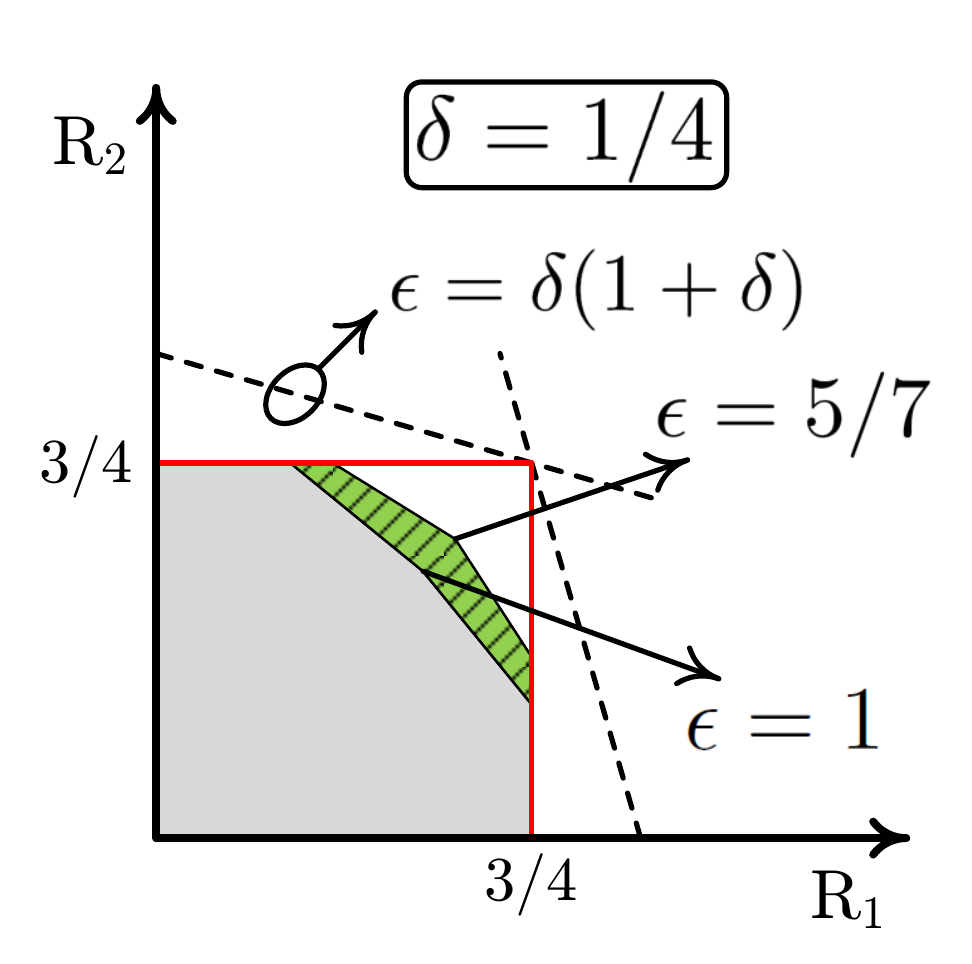}
\caption{The outer-bound region for $\dl = 1/4$ and different values of $\epsilon$.\label{Fig:Region-EIC}}
\end{figure}

Figure~\ref{Fig:Region-EIC} depicts the outer-bound region for $\dl = 1/4$ and different values of $\epsilon$. The gray shaded region is the baseline with no side-information ($\epsilon = 1$), and the hashed green region is the gain when $\epsilon = 5/7$. As we show in Theorem~\ref{THM:Achievability-EIC} and under the specified conditions, for $5/7 \leq \epsilon \leq 1$, we can achieve the outer-bounds and thus, the capacity region is characterized. We further note that for $\epsilon \leq \dl(1+\dl) = 5/16$, the outer-bound region is simply expressed by $R_i \leq (1 - \dl) = 3/4$.


\subsection{Achievable rates}
\label{Section:MainAch}

In this section, we consider a specific subset of the interference channels with altering topology in which a total of four topology configurations may occur at any given time as summarized in Figure~\ref{Fig:Topologies}--Topology $A$, $B$, $C$, and $D$, with respective probabilities:
\begin{align}
p_A = (1-\dl)^2,~p_B = \dl(1-\dl),~p_C = \dl(1-\dl),~p_D = \dl^2, 
\end{align}
and we further assume $\epsilon_1 = \epsilon_2 = \epsilon$. 
 
There are multiple reasons for choosing this specific channel distribution as in general with four binary links, a total of $16$ channel realizations would be possible as considered in~\cite{AlirezaBFICDelayed}. However, including all cases would make tracking the status of the previously transmitted signals more complicated; while the four channel realizations of Figure~\ref{Fig:Topologies} maintain the key technical challenges and simplify the analysis. Second, the channel realizations of Figure~\ref{Fig:Topologies} have an interesting motivation from two-unicast networks with a group of relays and the end-to-end network could be captured with these realizations~\cite{shomorony2012two,vahid2015informational,issa2015two,vahid2019degrees}. 

\begin{figure}[!ht]
\centering
\includegraphics[width = \columnwidth]{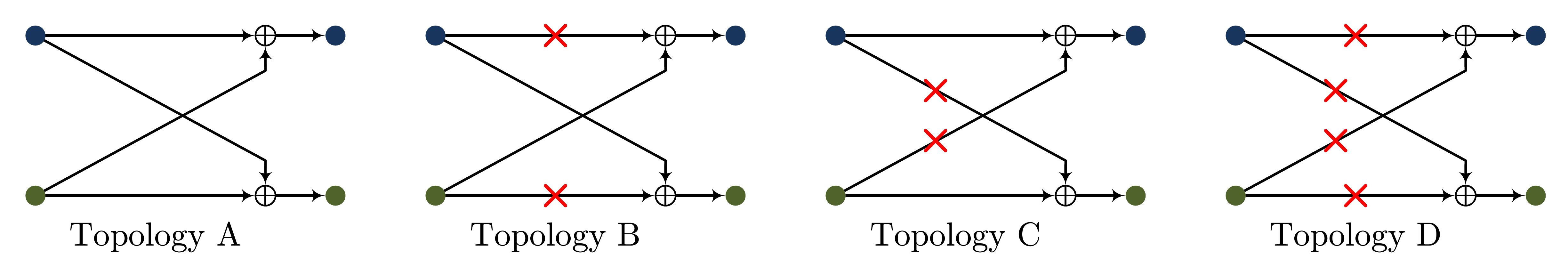}
\caption{For the inner-bounds, we assume the channel can only fall into one of four configurations at any given time.\label{Fig:Topologies}}
\end{figure}

We note that the outer-bounds of Theorem~\ref{THM:Capacity_Out_EICwCache} where derived for channels distributed independently across time but not necessarily across users. Thus, those outer-bounds hold for the channel we consider in this part as summarized in Figure~\ref{Fig:Topologies}, and can be re-written as
\begin{equation}
\label{Eq:OuterFourCases}
\left\{ \begin{array}{ll}
0 \leq R_i \leq \left( 1 - \dl \right), & i=1,2, \\
0 \leq \frac{\epsilon}{1+\dl} R_i + R_{\bar{i}} \leq \left( 1 - \dl^2 \right), & i = 1,2.
\end{array} \right.
\end{equation} 

The following theorem establishes the conditions under which the outer-bounds of Theorem~\ref{THM:Capacity_Out_EICwCache} are achievable.

\begin{theorem}[Achievability Conditions]
\label{THM:Achievability-EIC}
For the two-user interference channel with altering topology, delayed CSIT, and random receiver cache as described in Section~\ref{Section:Problem_EICwCache}, we have
\begin{enumerate}

\item {\bf Sum-Capacity}: The maximum sum-rate outer-bound of Theorem~\ref{THM:Capacity_Out_EICwCache}, specialized in \eqref{Eq:OuterFourCases} to the channel described above, is achievable when
\begin{align}
\label{Eq:conditionepsilon-EIC}
\epsilon \geq \frac{1}{1+\frac{(1 - 2\dl)^+}{1+\dl}}.
\end{align}

\item {\bf Capacity Region}: The entire outer-bound region of Theorem~\ref{THM:Capacity_Out_EICwCache}, specialized in \eqref{Eq:OuterFourCases} to the channel described above, is achievable when
\begin{align}
\label{Eq:conditionepsilon2-EIC}
\epsilon \geq \max\left\{ \frac{1}{1+\frac{(1 - 2\dl)^+}{1+\dl}}, \frac{\dl(1+\dl)}{(1-\dl)} \right\}.
\end{align}

\end{enumerate}
\end{theorem}

First, we note that for $\dl > 1/2$, the condition expressed in \eqref{Eq:conditionepsilon-EIC} implies $\epsilon = 1$, \emph{i.e.} no side-information at the receivers, which is covered in~\cite{AlirezaBFICDelayed}. Thus, we focus on $\dl \leq 1/2$, and \eqref{Eq:conditionepsilon-EIC} becomes:
\begin{align}
\label{Eq:conditionepsilon-EIC-revised}
\epsilon \geq \frac{1+\dl}{2 - \dl}, \quad \text{for~} \dl \leq \frac{1}{2}.
\end{align}
We further note that \eqref{Eq:conditionepsilon-EIC-revised} also implies that 
\begin{align}
\label{Eq:EpsilonLarger}
\epsilon \geq \dl (1+\dl), 
\end{align}
and based on the outer-bounds expressed in \eqref{Eq:RegionSym} and Remark~\ref{Remark:Threshold}, at the maximum sum-rate point, we have
\begin{align}
R_i = \frac{\left( 1 + \dl \right) \left( 1 - \dl^2 \right)}{1 + \dl + \epsilon}.
\end{align}
Finally, if $\dl \leq (+3-\sqrt{5})/2 \approx 0.382$ and the condition in \eqref{Eq:conditionepsilon-EIC} is satisfied, then \eqref{Eq:conditionepsilon2-EIC} also holds. In other words, for  $\dl \leq (+3-\sqrt{5})/2$ and $\epsilon$ satisfying \eqref{Eq:conditionepsilon-EIC}, the capacity region is characterized. 


\begin{figure}[!ht]
\centering
\includegraphics[width = 0.65\columnwidth]{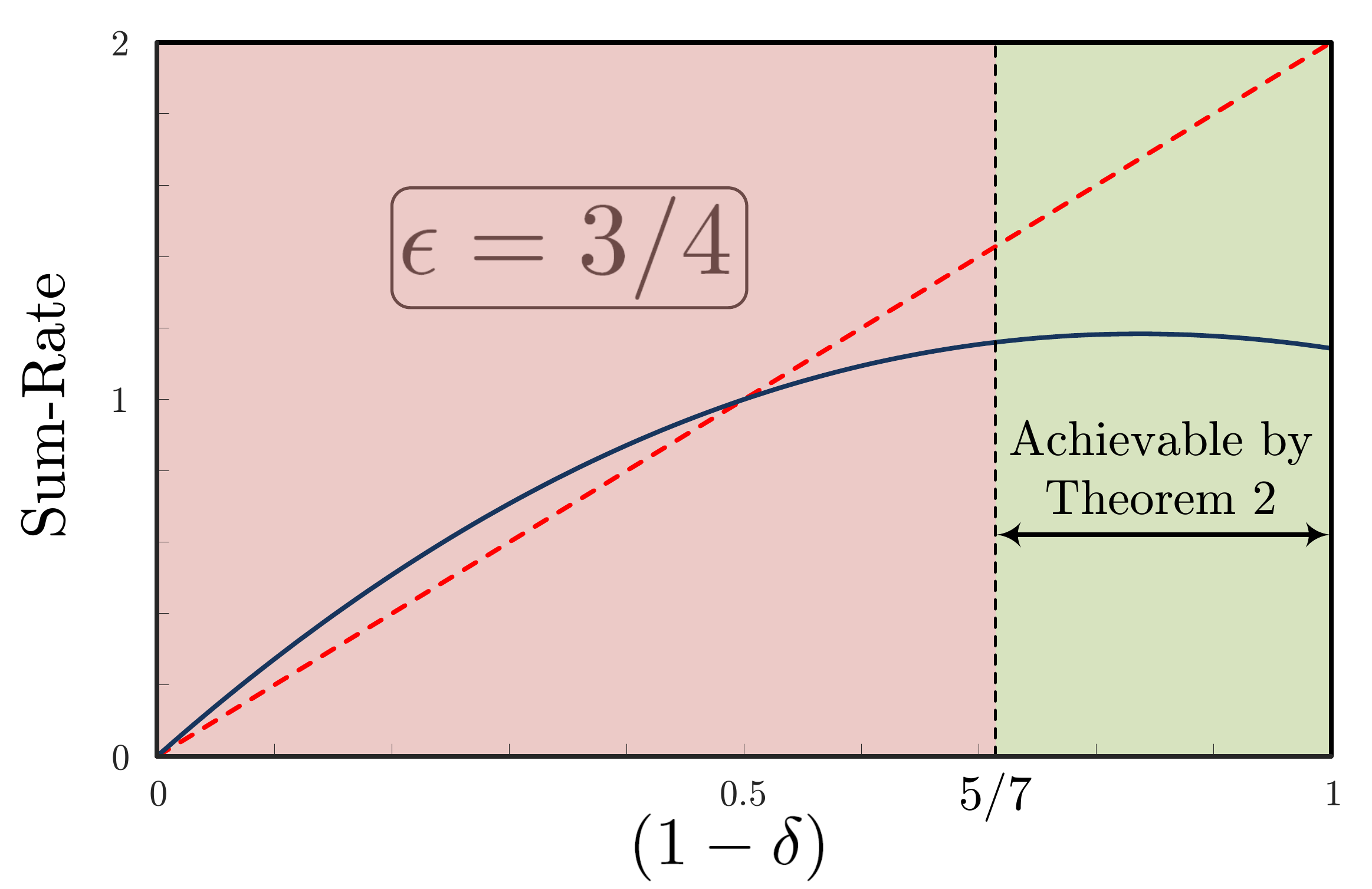}
\caption{Sum-rate outer-bound for the symmetric channel of Figure~\ref{Fig:Topologies} and the region in which the outer-bound is achievable.\label{Fig:Sum-Rates-Ach}}
\end{figure}

Figure~\ref{Fig:Sum-Rates-Ach} depicts the sum-rate outer-bound of Theorem~\ref{THM:Capacity_Out_EICwCache}, specialized in \eqref{Eq:OuterFourCases} to the channel we study in this section, for $\epsilon = 3/4$ and highlights the region in which the sum-rate outer-bound is achievable based on Theorem~\ref{THM:Achievability-EIC}.

\section{Proof of Theorem~\ref{THM:Capacity_Out_EICwCache}: Deriving the outer-bounds}
\label{Section:Converse_EIC}

In this section, we derive the outer-bounds of Theorem~\ref{THM:Capacity_Out_EICwCache}. The derivation of the outer-bounds on individual rates is straightforward, and thus, omitted. In fact, it is well established that the capacity of a point-to-point erasure channel with erasure probability $\delta_i$ is $\lp 1 - \delta_i \rp$. In what follows, we derive
\begin{align}
\label{Eq:ConverseBound}
\beta_1 R_1 + R_{2} \leq \left( 1 - \dl_{{\sf Rx}_2} \right),
\end{align}
where as indicated in \eqref{Eq:Beta_EICwCache}, we have 
\begin{align}
\label{Eq:Beta1_EICwCache}
\beta_1 = \frac{\epsilon_2 \left( 1 - \dl_{21} \right)}{\left( 1 - \dl_{{\sf Tx}_1} \right)}.
\end{align}
The derivation of the other bound would follow by simply changing user IDs $1 \leftrightarrow 2$.

Suppose rate-tupe $\lp R_1, R_2 \rp$ is achievable. We enhance receiver $\msf{Rx}_1$ by providing the entire $W_2$ to it, as opposed to $W_{2|1}$, and we note that this cannot reduce the rates. Then, we have
\begin{align}
n &\left( \beta_1 R_1 + R_2 \right) = \beta_1 H(W_1) + H(W_2) \nonumber \\
& \overset{(a)}= \beta_1 \underbrace{H(W_1|W_2, G^n)}_{\mathrm{Enhanced}~{\sf Rx}_1} + H(W_2|W_{1|2},G^n) \nonumber \\
& \overset{(\mathrm{Fano})}\leq \beta_1 I(W_1;Y_1^n|W_2,G^n) + I(W_2;Y_2^n|W_{1|2},G^n) + n \upxi_n \nonumber \\
& = \beta_1 H(Y_1^n|W_2,G^n) - \beta_1 \underbrace{H(Y_1^n|W_1,W_2,G^n)}_{=~0} \nonumber \\
& \quad + H(Y_2^n|W_{1|2},G^n) - H(Y_2^n|W_{1|2},W_2,G^n) + n \upxi_n \nonumber \\
& \overset{(b)}\leq H(Y_2^n|W_{1|2},G^n) + 2n \upxi_n \nonumber \\
& \overset{(c)}\leq n \left( 1 - \dl_{{\sf Rx}_2} \right) + 2\upxi_n,
\end{align}
where $\upxi_n \rightarrow 0$ as $n \rightarrow \infty$; $(a)$ follows from the independence of the messages and the channels, and captures the enhancement of receiver $\msf{Rx}_1$; $(b)$ follows from Theorem~\ref{THM:Leakage_BIC} below; $(c)$ is true since the entropy of a binary random variable is at most one, and the receiver is not in erasure a fraction $\lp 1 - \dl_{{\sf Rx}_2} \rp$ of the communication time. Dividing both sides by $n$ and let $n \rightarrow \infty$, we get \eqref{Eq:ConverseBound}.

\begin{theorem}
\label{THM:Leakage_BIC}
For the two-user interference channel with altering topology, delayed CSIT, and random receiver cache as described in Section~\ref{Section:Problem_EICwCache}, and $\beta_1$ given in (\ref{Eq:Beta1_EICwCache}), we have
\begin{align}
H\left( Y_2^n | W_{1|2}, W_2, G^n \right) + n \upxi_n \geq \beta_1 H\left( Y_1^n | W_2, G^n \right),
\end{align}
where $\upxi_n \rightarrow 0$ as $n \rightarrow \infty$.
\end{theorem}

\begin{proof}
For time instant $t$ where $1 \leq t \leq n$, we have
\begin{align}
& H\left( Y_2[t] | Y_2^{t-1}, W_{1|2}, W_2, G^n \right) \nonumber \\
& \quad \overset{(a)}= H\left( Y_2[t] | Y_2^{t-1}, W_{1|2}, W_2, G^t \right) \nonumber \\
& \quad \overset{(b)}= H\left( G_{21}[t]X_1[t] \oplus G_{22}[t]X_2[t]  | Y_2^{t-1}, W_{1|2}, W_2, G^t \right) \nonumber \\
& \quad \overset{(c)}= H\left( G_{21}[t]X_1[t] \oplus G_{22}[t]X_2[t]  | Y_2^{t-1}, X_2[t], W_{1|2}, W_2, G^t \right) \nonumber \\
& \quad = H\left( G_{21}[t]X_1[t] | Y_2^{t-1}, X_2[t], W_{1|2}, W_2, G^t \right) \nonumber \\
& \quad = \lp 1 - \delta_{21} \rp H\left( X_1[t] | Y_2^{t-1}, X_2[t], W_{1|2}, W_2, G_{21}[t] = 1,  G^{t-1} \right) \nonumber \\
& \quad \overset{(d)}= \lp 1 - \delta_{21} \rp H\left( X_1[t] | Y_2^{t-1}, X_2[t], W_{1|2}, W_2, G^t \right) \nonumber \\
& \quad \overset{(e)}\geq \lp 1 - \delta_{21} \rp H\left( X_1[t] | Y_1^{t-1},Y_2^{t-1}, X_2[t], W_{1|2}, W_2, G^t \right) \nonumber \\
& \quad \overset{(f)}= \frac{\lp 1 - \delta_{21} \rp}{\lp 1 - \dl_{{\sf Tx}_1} \rp} H\left( G_{11}[t]X_1[t], G_{12}[t]X_1[t] | Y_1^{t-1},Y_2^{t-1}, X_2[t], W_{1|2}, W_2, G^t \right) \nonumber \\
& \quad \overset{(g)}= \frac{\lp 1 - \delta_{21} \rp}{\lp 1 - \dl_{{\sf Tx}_1} \rp} H\left( Y_1[t], Y_2[t] | Y_1^{t-1},Y_2^{t-1}, X_2[t], W_{1|2}, W_2, G^t \right) \nonumber \\
& \quad \overset{(h)}= \frac{\lp 1 - \delta_{21} \rp}{\lp 1 - \dl_{{\sf Tx}_1} \rp} H\left( Y_1[t], Y_2[t] | Y_1^{t-1},Y_2^{t-1}, W_{1|2}, W_2, G^t \right) \nonumber \\
& \quad \overset{(i)}= \frac{\lp 1 - \delta_{21} \rp}{\lp 1 - \dl_{{\sf Tx}_1} \rp} H\left( Y_1[t], Y_2[t] | Y_1^{t-1},Y_2^{t-1}, W_{1|2}, W_2, G^n \right),
\end{align}
where $(a)$ follows from the temporal independence of the channels and the causal feedback structure;  $(b)$ follows from the channel model of \eqref{Eq:ReceivedSignal}; $(c)$ holds since according to \eqref{eq_enc_function_EICwCache0}, we have:
\begin{align}
X_2[t] = f_{2,t}\lp W_2, G^{t-1} \rp;
\end{align}
$(d)$ holds since according to the delayed CSIT assumption $X_1[t]$ is independent of the channel realizations at time instant $t$; $(e)$ follows from the fact that conditioning reduces entropy; $(f)$ comes from the fact that $\Pr\lp G_{11}[t] = G_{21}[t] = 0 \rp = \dl_{{\sf Tx}_1}$; $(g)$ follows from the fact that $X_2[t]$ and $G[t]$ are in the condition; $(h)$ holds for the same reason as step $(c)$; and $(i)$ follows for the same reason as step $(a)$.

Next, taking the summation over $t$ from $1$ to $n$, and using the fact that the transmit signal at time instant $t$ is independent of future channel realizations, we get
\begin{align}
\label{Eq:withSideInfo}
& H\left( Y_2^n | W_{1|2}, W_2, G^n \right) \geq \frac{\lp 1 - \delta_{21} \rp}{\lp 1 - \dl_{{\sf Tx}_1} \rp} H\left( Y_1^n, Y_2^n | W_{1|2}, W_2, G^n \right) \nonumber \\
& \quad \geq \frac{\lp 1 - \delta_{21} \rp}{\lp 1 - \dl_{{\sf Tx}_1} \rp} H\left( Y_1^n | W_{1|2}, W_2, G^n \right).
\end{align}
Now, for the final step, we note that
\begin{align}
\label{Eq:Deterministic}
H\left( Y_1^n | W_1, W_2, G^n \right) = 0,
\end{align}
and
\begin{align}
\label{Eq:Fano}
H\left( W_1 | Y_1^n, W_2, G^n \right) \leq n \upxi_n.
\end{align}
Thus, from \eqref{Eq:Deterministic} and \eqref{Eq:Fano}, we get
\begin{align}
\label{Eq:RemovingSideInfo}
H&\left( Y_1^n | W_{1|2}, W_2, G^n \right) \geq H\lp \bar{W}_{1|2} \rp - n \upxi_n \nonumber \\
& \overset{\eqref{Eq:MessageSplitting}}= \epsilon_2 H\lp W_1 \rp - n \upxi_n \overset{(a)} \geq \epsilon_2 H\lp Y_1^n | W_2, S^n \rp - n \upxi_n,
\end{align}
where $\bar{W}_{1|2}$ is the complement of $W_{1|2}$ in $W_1$, and $(a)$ is obtained from \eqref{Eq:Deterministic} as
\begin{align}
H\left( Y_1^n | W_1, W_2, G^n \right) = 0 \Rightarrow H\lp W_1 \rp \geq H\lp Y_1^n | W_2, G^n \rp. \nonumber
\end{align}
Finally, from \eqref{Eq:withSideInfo} and \eqref{Eq:RemovingSideInfo}, we obtain
\begin{align}
H&\left( Y_2^n | W_{1|2}, W_2, G^n \right) \overset{\eqref{Eq:withSideInfo}}\geq \frac{\lp 1 - \delta_{21} \rp}{\lp 1 - \dl_{{\sf Tx}_1} \rp} H\left( Y_1^n | W_{1|2}, W_2, G^n \right) \nonumber \\
& \overset{\eqref{Eq:RemovingSideInfo}}\geq \frac{\epsilon_2 \lp 1 - \delta_{21} \rp}{\lp 1 - \dl_{{\sf Tx}_1} \rp} H\left( Y_1^n | W_2, G^n \right) - n \upxi_n \nonumber \\
& \overset{\eqref{Eq:Beta1_EICwCache}}= \beta_1 H\left( Y_1^n | W_2, G^n \right) - n \upxi_n.
\end{align}
It is worth noting that a similar technique to what we used to obtain \eqref{Eq:RemovingSideInfo} was also used in~\cite{ghorbel2016content}. This completes the proof of Theorem~\ref{THM:Leakage_BIC}.
\end{proof}





\section{Proof of Theorem~\ref{THM:Achievability-EIC}: Exploiting knowledge of\\ receiver side-information at the Transmitter}
\label{Section:Achievability_EIC}


In this section, we provide the proof of Theorem~\ref{THM:Achievability-EIC}, \emph{i.e.} the achievability, and we show that the achievable region matches the outer-bounds of Theorem~\ref{THM:Capacity_Out_EICwCache} under the specified conditions. We first provide a motivating example and then, present the detailed proof. We also note that in this section, for the proofs and the motivating example, we assume the conditions and assumptions of the theorem are satisfied, and later in Section~\ref{Section:EIC_Discussion} we provide some intuitions for when these conditions are violated.

\subsection{Key ideas and motivating example}
\label{Section:KeyIdeas}

\noindent {\bf Key ideas}: The power of wireless is in multicast. More specifically, gains are magnified if we can satisfy multiple users simultaneously. This simple and intuitive idea is behind most wireless communication algorithms. To see how this idea can be used in interference channels with altering topology, delayed CSIT, and random receiver cache, we first explain the network coding opportunities, and then, present an example to explain the overall achievability strategy.

\begin{figure}[!ht]
\centering
\subfigure[]{\includegraphics[width = 0.43\columnwidth]{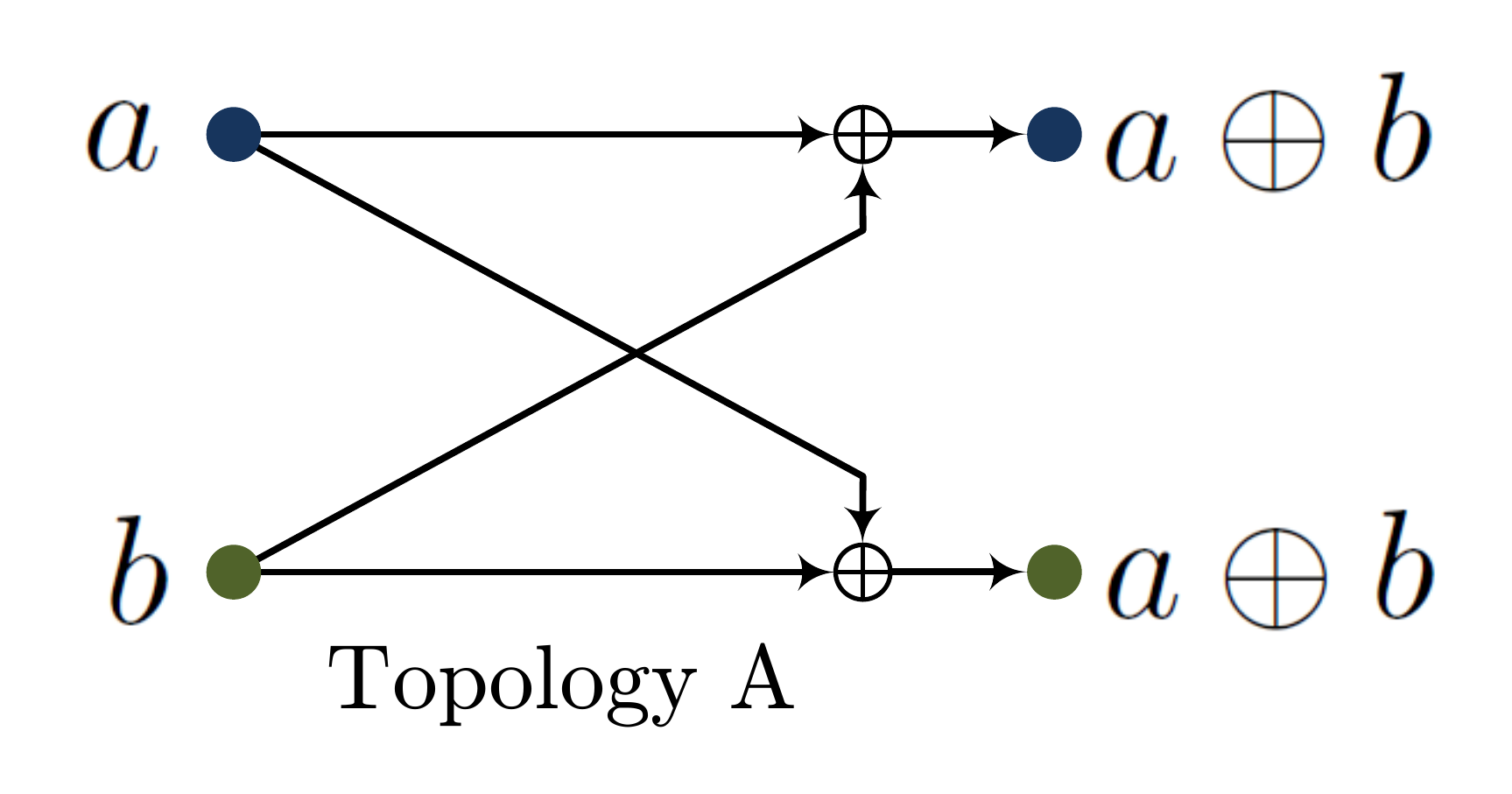}}
\hspace{.5in}
\subfigure[]{\includegraphics[width = 0.375\columnwidth]{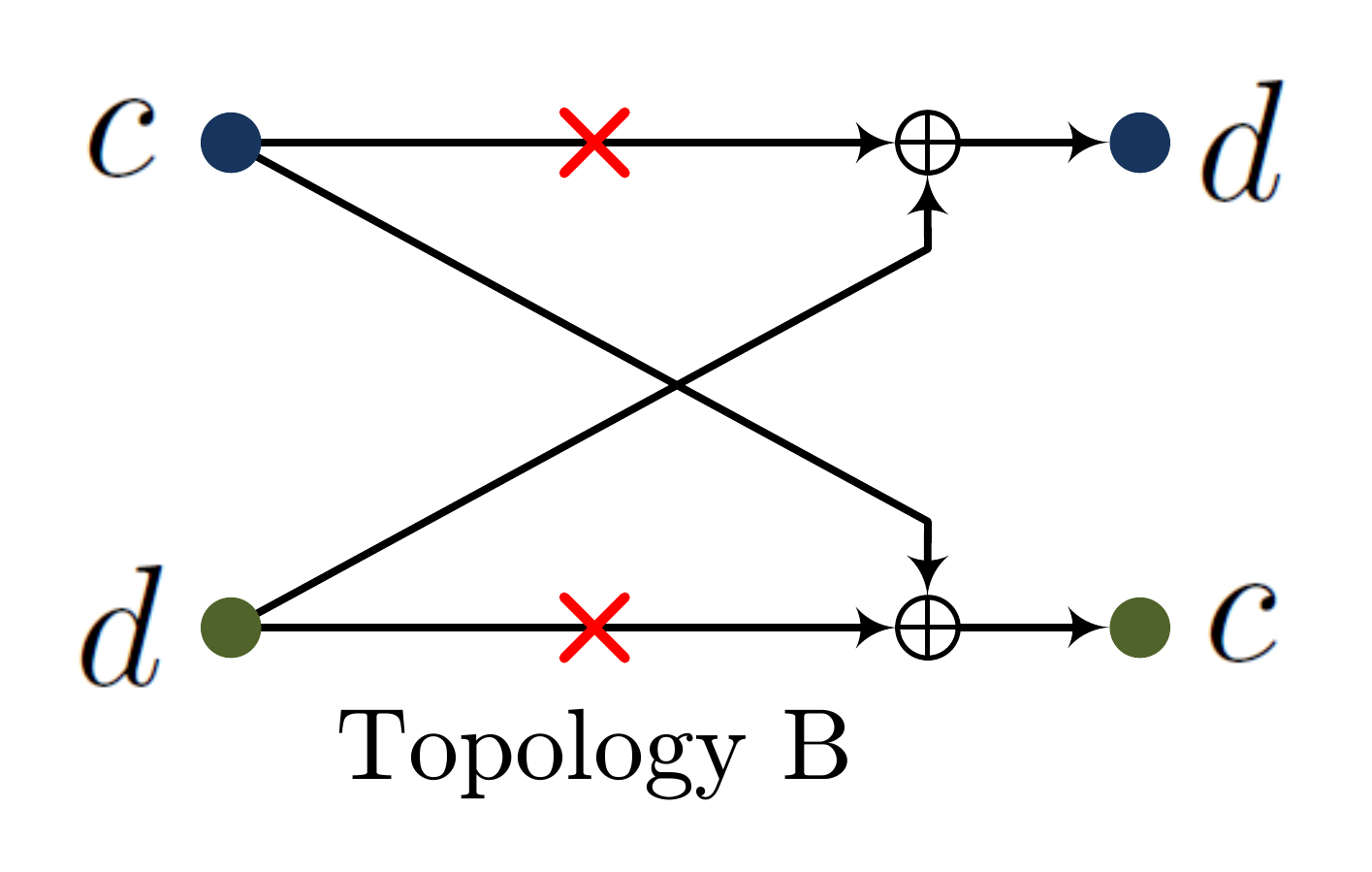}}
\caption{The transmitted bits in these examples can be combined for efficient multicast retransmission.\label{Fig:Cases-EIC}}
\end{figure}

Suppose at some time instant $t_1$, network topology $A$ is realized, meaning that all wireless links are active as in Figure~\ref{Fig:Cases-EIC}(a). If bits $a$ and $b$ were transmitted from $\msf{Tx}_1$ and $\msf{Tx}_2$, respectively, then each receiver obtains a linear combination of these bits\footnote{Note that if we assumed further continuous channel gains on top of the binary coefficients, these equations would have been almost surely linearly independent, which would be of value in general. However, by limiting ourselves to binary coefficients, the two equations are identical and need to make smart re-transmission choices in the future.}. It seems that providing only $a$ or $b$ to both receivers would be the optimal solution. However, we show there are other opportunities that can improve the network throughput. Suppose $c$ is part of $W_{1|2}$ (intended for $\msf{Rx}_1$ but in $\msf{Rx}_2$'s cache) and $d$ is part of $W_{2|1}$. This scenario could have also happened if at the time of transmitting $c$ and $d$, topology $B$ was realized as in Figure~\ref{Fig:Cases-EIC}(b). This time, it seems $c$ and $d$ could be sent to their respective receivers through two non-interfering point-to-point erasure links as they are apriori known to the unintended receiver. Interestingly, we can come up with a more efficient solution: $\msf{Tx}_1$ should deliver $a \oplus c$ and $\msf{Tx}_2$ should deliver $b \oplus d$ to {\it both} receivers. This way, $\msf{Rx}_1$ will end up with $a \oplus b, d, a \oplus c$, and $b \oplus d$, from which, it can recover $a$ and $c$. A similar story goes for $\msf{Rx}_2$. In summary, we achieved multicasting gains by mixing the signals available locally to each transmitter rather than retransmitting individual ones. 

\noindent {\bf Motivating example}: To keep the description short and convey the main points, for the motivating example, we focus on the maximum sum-rate point. We further use expected values of random variables as opposed to a more careful analysis involving concentration theorems and defer such analysis to the next subsection where we present the complete proof. We further choose an example where $\epsilon$ satisfies \eqref{Eq:conditionepsilon-EIC} of Theorem~\ref{THM:Achievability-EIC} with equality, which further shortens the description of achievability. In particular, we assume
\begin{align}
\dl = \frac{1}{5}, \epsilon = \frac{2}{3}.
\end{align}
This scenario corresponds to strong point-to-point erasure links (success rate of $4/5$) and when each receiver has apriori access to $1/3$ of the message of the other user. For these parameters, the maximum sum-rate point using \eqref{Eq:OuterFourCases} is:
\begin{align}
\left( R_1, R_2 \right) \approx \left( 0.62, 0.62 \right).
\end{align}
 
We start with $m$ bits for each receiver where $1/3m$ of the bits for each receiver is apriori known to the unintended receiver. Each transmitter separates its bits into two groups, the first are those known to the unintended receiver, called the side-information bits, and the second would be the complement of the first group. Each transmitter keeps sending out one bit from the second group until the channel realization learned through the feedback channel is {\emph not} topology $D$. This process on average takes
\begin{align}
\frac{\epsilon m}{1-\dl^2}. 
\end{align}
After this initial phase, each bit falls into three categories based on the topology that was realized during its transmission (topology $A$, $B$, or $C$). Those in topology $C$ are already delivered and no further action is needed. For those in topologies $A$ and $B$, we can retransmit the combination of them as discussed above. Further, for the choice of $\dl = 1/5$, there will be more bits associated with topology $A$ than $B$. We take advantage of this and mix the remaining bits of topology $A$ with the side-information bits at each transmitter, \emph{i.e.} those known apriori to the unintended receiver through the random cache. In this example, $\dl$ and $\epsilon$ were carefully chosen such that the number of bits in topology $A$ was exactly equal to those in topology $B$ and the random cache. The combined (XORed) bits can be delivered at the multiple-access channel (MAC) capacity formed at each receiver equal to $( 1 - \dl^2 )$. In summary, for this particular example, the total communication time is given by
\begin{align}
t_{\mathrm{total}} = \underbrace{\frac{\epsilon m}{1-\dl^2}}_{\text{initial~phase}} + \underbrace{\frac{\epsilon (1-\dl)^2 m}{(1-\dl^2)^2}}_{\text{multicasting~XORed~bits}} \approx 1.62,
\end{align}
which combined with the fact that each transmitter had $m$ bits to deliver, immediately implies the desired sum-rate.

\subsection{Proof of Theorem~\ref{THM:Achievability-EIC}} 

As discussed earlier in Section~\ref{Section:MainAch}, we first note that if $\dl > 1/2$, then the condition in~\eqref{Eq:conditionepsilon-EIC} alongside the fact that $0 \leq \epsilon \leq 1$, \emph{i.e.}
\begin{align}
\epsilon \geq \frac{1}{1+\frac{(1 - 2\dl)^+}{1+\dl}},
\end{align}
is simply equivalent to $\epsilon = 1$, or no side-information at either of the receivers. This case falls under the results of~\cite{AlirezaBFICDelayed}. Thus, in this section we assume $\dl \leq 1/2$, and \eqref{Eq:conditionepsilon-EIC} becomes
\begin{align}
\label{Eq:NewEpsilon-EIC}
\epsilon \geq \frac{1+\dl}{2 - \dl}.
\end{align}

\begin{figure}[!ht]
\centering
\includegraphics[width = 0.45\columnwidth]{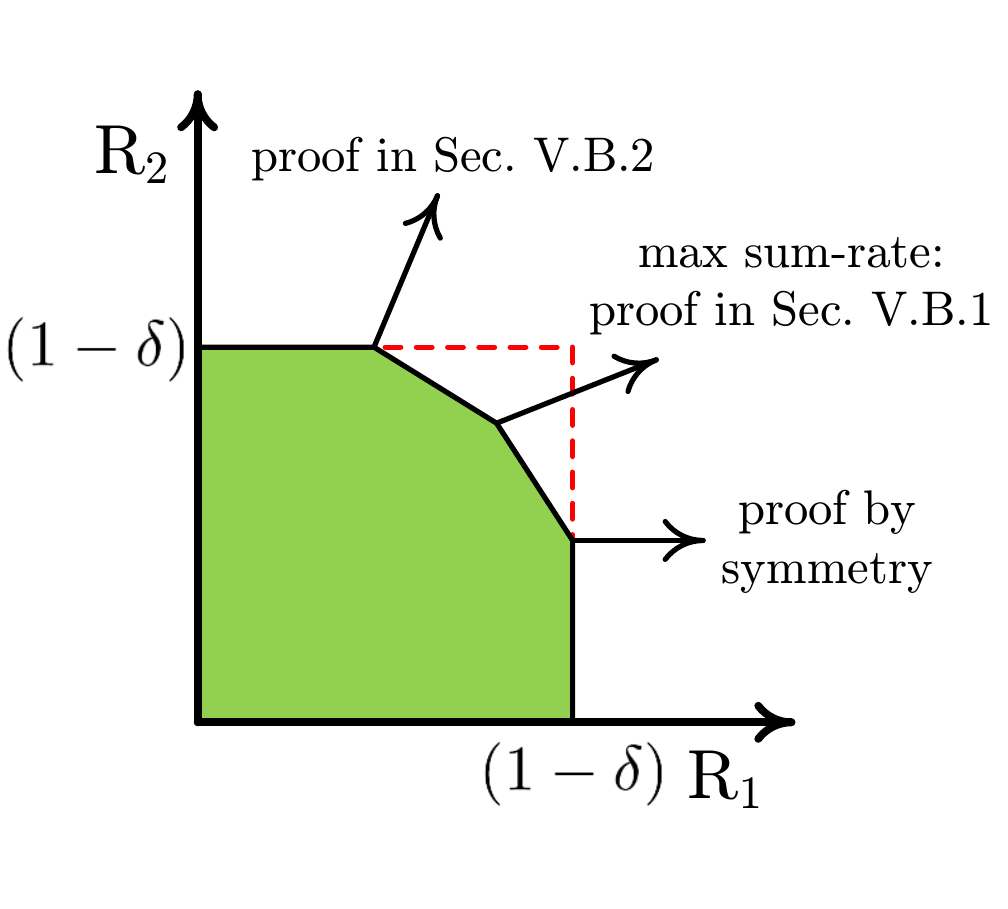}
\caption{The overall shape of the outer-bound region for $\dl \leq 1/2$ and when \eqref{Eq:conditionepsilon-EIC} is satisfied.\label{Fig:Region-Ach}}
\end{figure}

For this setting, all outer-bounds in Theorem~\ref{THM:Capacity_Out_EICwCache} are active and the region is depicted in Figure~\ref{Fig:Region-Ach}. To prove Theorem~\ref{THM:Achievability-EIC}, then it suffices to show the achievability of the maximum sum-rate point, \emph{i.e.}
\begin{align}
\label{Eq:MaxSumRate-EIC}
\left( R_1, R_2 \right) = \left( \frac{(1+\dl)(1-\dl^2)}{1+\dl+\epsilon}, \frac{(1+\dl)(1-\dl^2)}{1+\dl+\epsilon} \right),
\end{align}
and either of the corner points (the other would follow due to symmetry), \emph{e.g.},
\begin{align}
\label{Eq:MaxR2-EIC}
\left( R_1, R_2 \right) = \left( \frac{\dl(1+\dl)}{\epsilon} (1-\dl), (1-\dl) \right).
\end{align}
We remind the reader that according to Theorem~\ref{THM:Achievability-EIC}, we can achieve this latter corner-point when further \eqref{Eq:conditionepsilon2-EIC} is satisfied.

\subsubsection{Maximum sum-rate point}

We start with $m$ bits for each user out of which a fraction $(1-\epsilon)$ is available to the unintended receiver for $\epsilon$ satisfying \eqref{Eq:NewEpsilon-EIC}. As mentioned earlier, we refer to the bits known apriori to the receivers as ``side-information'' bits. The achievability strategy is divided into three phases. 

The first phase is uncoded, uncategorized transmission and each bit that is not known apriori to the unintended receiver, is repeated until the channel realization is not topology $D$ of Figure~\ref{Fig:Topologies}. In other words, each bit is repeated as long as all channels are equal to $0$. After this phase, three categories can be identified: delivered bits (topology $C$), side-information bits (those in Topology $B$ and the ones known apriori to the transmission at the receivers), and bits in topology $A$. 

Phase~2 benefits from the idea described in Section~\ref{Section:KeyIdeas} and XORs all the side-information bits with those in topology $A$, and sends the resulting bits at the multicast sum-rate of $(1-\dl^2)$. As we show shortly,  \eqref{Eq:NewEpsilon-EIC} guarantees there will be more bits (on average) in topology $A$ than side-information bits. 

In Phase~3, the remaining bits in topology $A$ need to be delivered. However, as mentioned in Section~\ref{Section:KeyIdeas}, as there are no more network coding opportunities available, it suffices to send half of these bits at the multicast sum-rate. For example, in the scenario of Figure~\ref{Fig:Cases-EIC}(a), delivering only $a$ or $b$ to both receivers is sufficient for decoding. We note that if \eqref{Eq:NewEpsilon-EIC} holds with equality, the communication ends after Phase~2 and Phase~3 is not needed.

\noindent {\bf Rate analysis}: Phase~1 takes a total time of
\begin{align}
t_{\mathrm{P.1}} = \frac{\epsilon m}{(1-\dl^2)} + \mathcal{O}\left( m^{2/3}\right),
\end{align}
where the big $O$ notation is used in its standard definition. The reason for the addition of $\mathcal{O}\left( m^{2/3}\right)$ is ensure all bits are communicated with probability approaching $1$ as $m \rightarrow \infty$. If at the end of Phase~1, there are still some bits left for transmission at either of the transmitters, we declare an error of type-I and terminate the communication.

Upon completion of Phase~1, denote the (random) number of bits at transmitter $\msf{Tx}_i$ in topology $A$ with $N_i^A$ and in topology $B$ with $N_i^B$. If either of the following inequalities hold, we declare an error of type-II and terminate the communication:
\begin{align}
& N_i^A < \frac{\Pr\left( \text{~topology~}A \right)}{1-\Pr\left( \text{~topology~}D \right)}\epsilon m - \mathcal{O}\left( m^{2/3}\right) = \frac{(1-\dl)^2}{(1-\dl^2)} \epsilon m -  \mathcal{O}\left( m^{2/3}\right), \nonumber\\
& N_i^B < \frac{\Pr\left( \text{~topology~}B \right)}{1-\Pr\left( \text{~topology~}D \right)}\epsilon m - \mathcal{O}\left( m^{2/3}\right) = \frac{\dl(1-\dl)}{(1-\dl^2)} \epsilon m -  \mathcal{O}\left( m^{2/3}\right). 
\end{align}

Phase~2 takes a total time of
\begin{align}
t_{\mathrm{P.2}} =  \frac{2}{(1-\dl^2)} \left( \underbrace{(1 - \epsilon)m}_{\text{apriori~side~info.}} + \underbrace{\frac{\dl(1-\dl)\epsilon m}{(1-\dl^2)}}_{\text{Topology~B}} \right) + \mathcal{O}\left( m^{2/3}\right),
\end{align}
 and if upon termination of Phase~2, there are still some bits left for transmission at either of the transmitters, we declare an error of type-I and terminate the communication.

If \eqref{Eq:NewEpsilon-EIC} holds with equality, the communication ends after Phase~2 and Phase~3 is not needed. Phase~3 takes a total time of
\begin{align}
t_{\mathrm{P.3}} =  \frac{1}{(1-\dl^2)} \left( \underbrace{\frac{(1-\dl)^2\epsilon m}{(1-\dl^2)}}_{\text{topology~}A} - \underbrace{(1 - \epsilon)m}_{\text{apriori~side~info.}} - \underbrace{\frac{\dl(1-\dl)\epsilon m}{(1-\dl^2)}}_{\text{Topology~B}} \right) + \mathcal{O}\left( m^{2/3}\right),
\end{align}
 and if upon termination of Phase~3, there are still some bits left for transmission at either of the transmitters, we declare an error of type-I and terminate the communication.

Using Bernstein inequality~\cite{bernstein1924modification}, we can show that the probability of errors of types I, II, and III decreases exponentially and approaches zero\footnote{In simple terms this inequality states that if $X_1,\ldots,X_r$ are $r$ independent random variables, and $M=\sum_{i=1}^r{X_i}$, then $Pr\left[ |M-\mathbb{E}\left[ M \right] | > \alpha \right] \leq 2 \exp \left( \frac{-\alpha^2}{4 \sum_{i=1}^r \mathrm{Var} \left( X_i \right)} \right)$.}  as $m \rightarrow \infty$. In fact, throughout this section, we intentionally picked $m^{\frac{2}{3}}$ to add to the constants in order to guarantee that the error terms vanish as $m$ increases. For instance, to bound the probability of error type-I, we have:
\begin{align}
& \Pr \left[ \mathrm{error~type \noindent - \noindent I} \right] \overset{\text{Union~Bound}}\leq  \sum_{i=1}^2{\Pr \left[ Q_{i \rightarrow i} \mathrm{~is~not~empty} \right]} \nonumber \\
&~= 4 \exp \left( \frac{-m^{4/3}}{4 \left( 1 - \dl^4 \right) \dl^4 \left[ \frac{1}{1-\dl^2} m + m^{\frac{2}{3}} \right]} \right),
\end{align}
which decreases exponentially to zero as $m \rightarrow \infty$. We refer the reader to~\cite{AlirezaBFICDelayed} for a more detailed discussion on using Bernstein inequality (and Chernoff-Hoeffding bound~\cite{Chernoff,Hoeffding}) to analyze the error probabilities. 

The total communication time can be calculated as follows:
\begin{align}
t_{\mathrm{total}} =t_{\mathrm{P.1}} + t_{\mathrm{P.2}} + t_{\mathrm{P.3}} = \frac{(1-\dl)}{(1-\dl^2)^2}\left( 1 + \dl + \epsilon \right) m + \mathcal{O}\left( m^{2/3}\right),
\end{align}
and as we started with a total of $2m$ bits, the maximum sum-rate corner-point given in \eqref{Eq:MaxSumRate-EIC} is achieved when $m \rightarrow \infty$.

\subsubsection{Maximum sum-rate when $R_2 = (1-\dl)$}

For this corner-point, based on the rates given in \eqref{Eq:MaxR2-EIC}, we start with $m$ bits for $\msf{Tx}_2$, and 
\begin{align}
\label{Eq:m1-EIC}
m_1 = \frac{\dl(1+\dl)}{\epsilon} m
\end{align}
bits for $\msf{Tx}_1$, where based on the choice parameters described in this section, as shown in \eqref{Eq:EpsilonLarger}, we have $m_1 < m$. This section will also demonstrate the challenges in achieving the outer-bounds when rates are unequal. The main idea is for $\msf{Tx}_1$ to be responsible for interference management, while $\msf{Tx}_2$ communicates at the maximum possible rate.

In Phase~1, $\epsilon m_1$ bits from each transmitter will be communicated similar to the maximum sum-rate case. In Phase~2, $\msf{Tx}_2$ will finish communicating the remainder of the bits that are not known apriori to $\msf{Tx}_1$, meanwhile, $\msf{Tx}_1$ delivers a portion of the bits that fell under topology $A$ in Phase~1 at a low enough rate, $\dl(1-\dl)$, such that both receivers can decode them. This way, $\msf{Tx}_1$ helps with interference management as $\msf{Tx}_2$ continues its communications. In Phase~3, $\msf{Tx}_2$ communicates the side-information bits (those known to $\msf{Tx}_1$), and meanwhile $\msf{Tx}_1$ encodes the remaining bits of topology $A$ at $\dl(1-\dl)$, its side-information bits at $(1-\dl)^2$, and sends the superposition of the resulting encoded bits. Figure~\ref{Fig:CornerPoint-Ach} pictorially summarizes the achievability strategy for this corner-point.

\begin{figure}[!ht]
\centering
\includegraphics[width = 0.75\columnwidth]{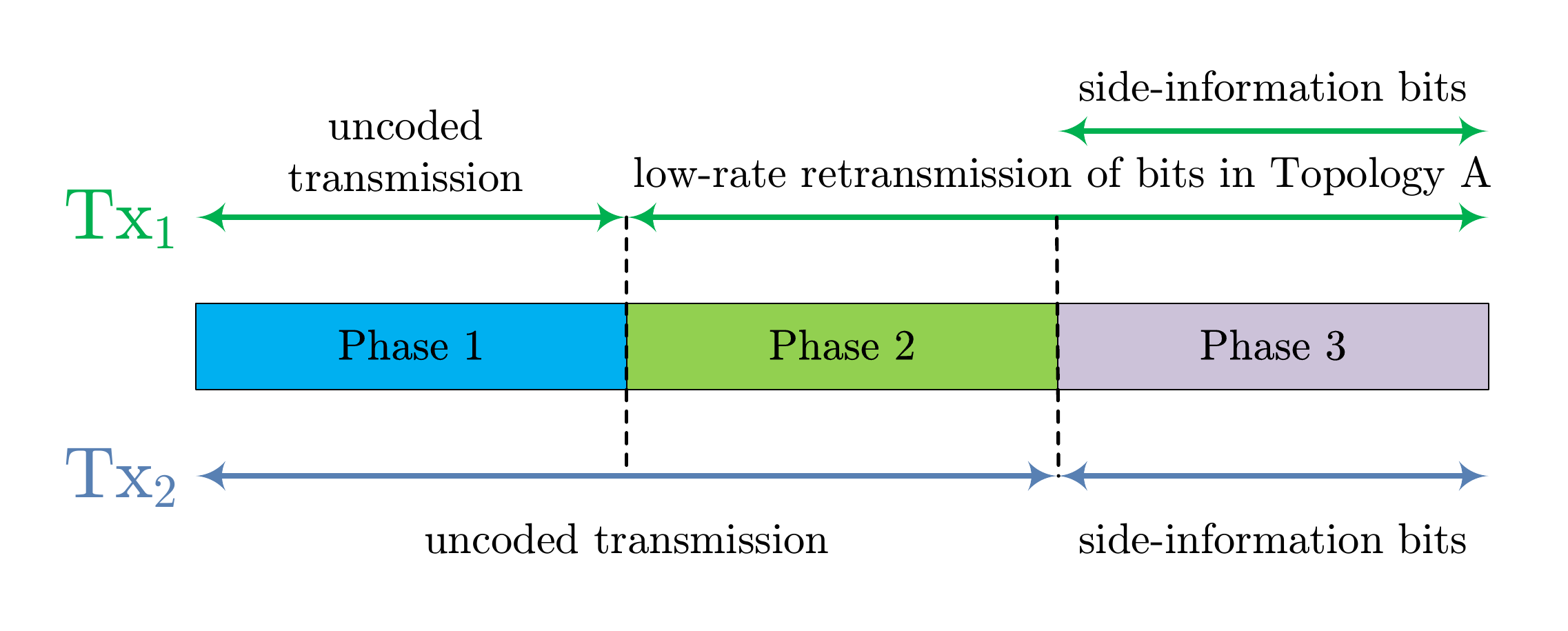}
\caption{A summary of the achievability strategy for the maximum sum-rate when $R_2 = (1-\dl)$.\label{Fig:CornerPoint-Ach}}
\end{figure}


\noindent {\bf Rate analysis}: For simplicity, we eliminate the $\mathcal{O}\left( m^{2/3}\right)$ terms and the error analysis as they follow the same steps as the previous case. Phase~1 takes (on average) a total time of
\begin{align}
t_{\mathrm{P.1}} = \frac{\epsilon m_1}{(1-\dl^2)} \overset{\eqref{Eq:m1-EIC}}= \frac{\dl}{(1-\dl)}m .
\end{align}

It is helpful to view Phases~2 and~3 from the perspective of $\msf{Tx}_1$. Essentially, during these two phases, $\msf{Tx}_1$ needs to deliver all its bits that fell under topology $A$ at a rate of $\dl (1-\dl)$. This rate captures when each receiver observes only the signal of $\msf{Tx}_1$. We will describe the decoding at each receiver is detail shortly, but based on the discussion, we have:
\begin{align}
t_{\mathrm{P.2}} + t_{\mathrm{P.3}} = \frac{1}{\dl(1-\dl)} \times \frac{(1-\dl)^2 \epsilon m_1}{(1-\dl^2)} = m,
\end{align}
which would result in:
\begin{align}
t_{\mathrm{total}} =t_{\mathrm{P.1}} + t_{\mathrm{P.2}} + t_{\mathrm{P.3}} \overset{\eqref{Eq:m1-EIC}}= \frac{m}{(1-\dl)}.
\end{align}
Given that $\msf{Tx}_2$ has $m$ bits and $\msf{Tx}_1$ has $m_1$ bits, given in \eqref{Eq:m1-EIC}, (assuming successful delivery which will be discussed below) the rates in \eqref{Eq:MaxR2-EIC} will be achieved for $m \rightarrow \infty$.

To show decodability of all bits, we first note that all the bits of $\msf{Tx}_1$ that fell under topology $A$ in Phase~1 are decodable at both receivers given the chosen rate of $\dl (1-\dl)$. In other words, $\msf{Tx}_2$ will never retransmit any of its bits under topology $A$ as it must deliver its bits at maximum point-to-point rate of $(1-\dl)$. In particular, during Phase~2, we can assume  $\msf{Rx}_2$ does not receive any interference as it can decode and subtract the contribution from $\msf{Tx}_1$.

Phase~3 is dedicated to the delivery of the side-information bits of $\msf{Tx}_2$, which will not cause any interference at $\msf{Rx}_1$, and meanwhile, $\msf{Tx}_1$ sends the superposition of the bits in topology $A$ and its side-information bits. The side-information bits do not cause any interference at $\msf{Rx}_2$, and the bits under topology $A$ are decodable at both receivers as discussed earlier. Thus, in Phase~3, we can assume interference-free transmission for both users. Phase~3 takes on average:
\begin{align}
t_{\mathrm{P.3}} = \frac{1}{(1-\dl)}\left( 1 - \frac{\epsilon}{(1+\dl)}\right),
\end{align}
and as long as this time is sufficient for the delivery of the side-information bits of $\msf{Tx}_1$ at rate $(1-\dl)^2$, the decodability is guaranteed.

The average number of the side-information bits at $\msf{Tx}_1$ is:
\begin{align}
\frac{\dl \left( 1 + \dl - \epsilon \right)}{\epsilon},
\end{align}
which at rate $(1-\dl)^2$, takes: 
\begin{align}
\frac{\dl \left( 1 + \dl - \epsilon \right)}{\epsilon (1-\dl)^2}
\end{align}
times to deliver. Thus, for this time to be less than or equal to $t_{\mathrm{P.3}}$, we require:
\begin{align}
\frac{\dl \left( 1 + \dl - \epsilon \right)}{\epsilon (1-\dl)^2} \leq t_{\mathrm{P.3}},
\end{align}
which in turn is equivalent to:
\begin{align}
\epsilon \geq \frac{\dl (1+\dl)}{(1-\dl)}.
\end{align}
This last inequality is what was required in Theorem~\ref{THM:Achievability-EIC} beyond \eqref{Eq:NewEpsilon-EIC} to guarantee the entire region would be achievable. Thus, the proof is completed.

\subsection{Discussion: when the conditions of Theorem~\ref{THM:Achievability-EIC} are violated}
\label{Section:EIC_Discussion}

To better understand the challenges when the conditions of Theorem~\ref{THM:Achievability-EIC} are violated, we provide an example in this section in which $\dl = 1/2$ and $\epsilon = 3/4$ (\emph{i.e.} only $1/4$ of each message is available to the unintended user). For these values, the outer-bound region of Theorem~\ref{THM:Capacity_Out_EICwCache} is equivalent to $R_i \leq 1/2$, $i = 1,2$. If these bounds are tight, the transmitters should be able to effectively eliminate all interference at the unintended users and achieve a sum-rate of $1$.

We focus on the maximum sum-rate and start with $m$ bits at each transmitter, out of which $1/4$ are side-information bits. The challenge is after the initial phase, as the four realizations in Figure~\ref{Fig:Topologies} are equiprobable for $\dl = 1/2$, the number of bits under topology $B$ equals that of the bits in topology $A$. Thus, after combining the bits, the side-information bits cannot be combined with any other bits and should be sent at the maximum point-to-point rate (\emph{i.e.} losing the multicast gain). The overall communication time for this example is given by:
\begin{align}
t_{\mathrm{total}} = \underbrace{\frac{1}{1-.5^2} \times \frac{3m}{4}}_{\text{Phase~1}} +  \underbrace{\frac{1}{1-.5^2} \times \frac{.5^2}{1-.5^2} \times \frac{3m}{4}}_{\text{Phase~2~multicast}} +  \underbrace{\frac{2}{1-.5} \times \frac{1m}{4}}_{\text{Phase~3~point-to-point}} = \frac{9m}{4}.
\end{align}
Thus, the achievable sum-rate is equal to $8/9$, which is below the sum-rate of $1$ the outer-bounds suggest. Rather surprisingly, this rate is even below the scenario with no side-information ($\dl = 1/2$ and $\epsilon = 1$) for which achievability of a sum-rate of $9/10$ was already derived in~\cite{AlirezaBFICDelayed}. In other words, in this particular example for our achievability scheme, it is better to ignore the side-information, suggesting that the inner-bounds need improvement. However, it seems no further linear coding opportunities are available to the transmitters. Therefore, it might be the case that to achieve the outer-bounds, one may need to exploit non-linear codes.

\section{Conclusion}
\label{Section:Conclusion_EIC}

In this paper, we studied the benefit of having random receiver cache in interference channels with altering topology and delayed feedback. We provided a new set of outer-bounds based on a key theorem that quantifies the baseline entropy available to each receiver under the specific assumptions of the problem. We showed that these bounds are tight under certain conditions, thus, characterizing the capacity region in such cases. We further discussed how the inner and the outer bounds behave when these conditions are violated. The next steps include investigating whether non-linear coding may improve the inner-bounds or the outer-bounds need improvement when the capacity remains open. Further, it would interesting to understand the implications of random receiver cache on latency and age of information in interference channels with altering topology.


\bibliographystyle{ieeetr}
{\footnotesize \bibliography{bib_FBBudget,bibTWR}}

\end{document}